\documentclass[journal]{IEEEtran}
\usepackage[cmex10]{amsmath}
\usepackage[mathcal]{euscript}

\usepackage{graphicx}
\usepackage{amssymb}
\usepackage{epstopdf}
\usepackage{amsthm}
\usepackage{enumerate}
\usepackage{eufrak}
\usepackage{cite}
\usepackage{mathcomp}
\usepackage{supertabular}
\usepackage{longtable}
\usepackage{stmaryrd}
\usepackage{url}
\usepackage{color}
\usepackage{rotating}
\usepackage{float}

\theoremstyle{definition}

\newtheorem{prop}{Proposition}[section]
\newtheorem{thm}{Theorem}[section]
\newtheorem{cor}{Corollary}[section]
\newtheorem{lem}{Lemma}[section]

\newtheorem{defn}{Definition}[section]

\newcommand \bfs {\mathbf{s}}
\newcommand \bfu {\mathbf{u}}
\newcommand \bfv {\mathbf{v}}
\newcommand \bfw {\mathbf{w}}
\newcommand \bfx {\mathbf{x}}

\newcommand \remove[1] {}

\begin{document}

\title{Cross-Bifix-Free Codes Within a \\ Constant Factor of Optimality}

\author{Yeow Meng Chee,~\IEEEmembership{Senior Member, IEEE},
Han Mao Kiah,
\\ Punarbasu Purkayastha,~\IEEEmembership{Member, IEEE},
and Chengmin Wang
\thanks{Research of Y.~M.~Chee, H.~M.~Kiah, 
and C.~Wang is supported in part by the National Research Foundation of Singapore under Research Grant NRF-CRP2-2007-03. C.~Wang is also supported in part by NSFC under Grant No.10801064 and 11271280.}

\thanks{Y.~M.~Chee, H.~M.~Kiah and P.~Purkayastha 
are with the Division~of~Mathematical Sciences,
  School~of~Physical~and~Mathematical~Sciences,
  Nanyang~Technological~University, 21~Nanyang~Link, Singapore~637371,
  Singapore (emails:\{ymchee, kiah0001, punarbasu\}@ntu.edu.sg).}%
\thanks{C.~Wang is with School~of~Science, Jiangnan~University, Wuxi 214122, China (email:chengmin\_wang09@yahoo.com.cn).}

}



\newcommand{\A}{{\cal A}}
\newcommand{\B}{{\cal B}}
\newcommand{\C}{{\cal C}}
\newcommand{\D}{{\cal D}}
\newcommand{\f}{{\cal F}}
\newcommand{\RS}{{\cal RS}}
\newcommand{\G}{{\cal G}}
\newcommand{\SSS}{{\cal S}}
\newcommand{\vc}{{\sf c}}
\newcommand{\vu}{{\sf u}}
\newcommand{\vs}{{\sf s}}

\newcommand{\vA}{{\sf A}}
\newcommand{\vB}{{\sf B}}
\newcommand{\vC}{{\sf C}}
\newcommand{\vR}{{\sf R}}
\newcommand{\tA}{\textrm A}
\newcommand{\tB}{\textrm B}

\newcommand{\CC}{\mathbb C} 
\newcommand{\RR}{\mathbb R}
\newcommand{\ZZ}{\mathbb Z}
\newcommand{\ff}{\mathbb Z}
\newcommand{\FF}{\mathbb F}
\newcommand{\nonneg}{\ZZ_{\ge 0}}
\newcommand{\ceiling}[1]{\left\lceil{#1}\right\rceil}
\newcommand{\floor}[1]{\left\lfloor{#1}\right\rfloor}
\newcommand{\mmod}{\textrm{ mod }}

\newcommand{\wt}[1]{\textrm{wt}{(#1)}}
\newcommand{\trace}[1]{\textrm{Trace}{(#1)}}
\newcommand{\supp}[1]{\textsf{supp}{(#1)}}
\newcommand{\lev}[1]{\textsf{lev}{(#1)}}
\newcommand{\dist}{\textsf{dist}}
\newcommand{\packing}[1]{\textrm{Packing}_\textrm{#1}}
\newcommand{\ppty}[1]{\textsf{Property #1}}

\newcommand{\block}{\mathcal B}
\newcommand{\group}{\mathcal G}

\newcommand{\gbtd}{\textrm{GBTD}}
\newcommand{\hgbtd}{\textrm{hGBTD}}

\newcommand{\bibd}{\textrm{BIBD}}
\newcommand{\rbibd}{\textrm{RBIBD}}

\newcommand{\kts}{\textrm{KTS}}
\newcommand{\fkts}{\textrm{FKTS}}

\newcommand{\td}{\textrm{TD}}
\newcommand{\drtd}{\textrm{DRTD}}

\newcommand{\gdd}{\textrm{GDD}}
\newcommand{\fgdrp}{\textrm{FGDRP}}
\newcommand{\pbd}{\textrm{PBD}}
\newcommand{\spec}{\textrm{Spec}}

\newcommand{\inprod}[1]{\langle{#1}\rangle}

\newcommand{\edit}[1]{{\color{red}#1}}
\newcommand{\edittwo}[1]{{\color{blue}#1}}
\newcommand{\citereq}{{\color{blue}[citation required]}}
\newcommand \etal{\emph{~et~al.~}}
\newcommand{\vvspace}{\vspace{1mm}}


\newcommand{\beas}{\begin{eqnarray*}}
\newcommand{\eeas}{\end{eqnarray*}}

\newcommand{\bm}[1]{{\mbox{\boldmath $#1$}}} 

\newcommand{\tworow}[2]{\genfrac{}{}{0pt}{}{#1}{#2}}
\newcommand{\qbinom}[2]{\left[ {#1}\atop{#2}\right]_q}

\newcommand{\Lovasz}{Lov\'{a}sz }
\newcommand{\pp}{^\prime}
\newcommand{\aaa}{^\ast}
\newcommand{\union}{\bigcup\limits}

\newcommand{\cS}{\mathcal{S}}   

\makeatletter
\def\imod#1{\allowbreak\mkern10mu({\operator@font mod}\,\,#1)}
\makeatother


\maketitle

\begin{abstract}
A cross-bifix-free code is a set of words in which no prefix of any length
of any word is the suffix of any word in the set. Cross-bifix-free codes
arise in the study of distributed sequences for frame synchronization. We
provide a new construction of cross-bifix-free codes which generalizes the
construction in Bajic (2007) to longer code lengths and to any alphabet
size. The codes are shown to be nearly optimal in size. We also establish new
results on Fibonacci sequences, that are used in estimating the size of the
cross-bifix-free codes.
\end{abstract}



\section{Introduction}

A crucial requirement to reliably transmit information
in a digital communication system is to establish synchronization between
the transmitter and the receiver. Synchronization is required not only to
determine the start of a symbol, but also to determine the start of a frame
of data in the received signals. The initial acquisition of frame
synchronization and the maintenance of this synchronization has been
a widely studied field of research for several decades. Early works on
frame synchronization
concentrated on introducing a synchronization word periodically into the
data stream \cite{mas72,nie73}. In the receiver, correlation techniques were used to
determine the position of the synchronization sequence within the data
stream. Massey \cite{mas72} introduced the notion of \emph{bifix-free}
synchronization word in order to achieve fast and reliable synchronization
in binary data streams. A bifix-free word denotes a sequence of symbols in
which no prefix of any length of the  word is identical to any suffix of
the word.

The current methods for achieving frame synchronization at the receiver do
not look at exact matching of the synchronization word. Instead, the
objective is to search for a word that is within a specified Hamming
distance of the transmitted synchronization word. This procedure allows for
faster synchronization between the transmitter and the receiver \cite{baj04}.
Van Wijngaarden and Willink\cite{wij00} introduced the notion of a \emph{distributed sequence} where the
synchronization word is not a contiguous sequence of symbols but is instead
interleaved into the data stream.
For instance the binary sequence $110*0*0$ is a distributed sequence where
the symbol $*$ denotes a data symbol that can take either of the values
0 or 1.
Van Wijngaarden and Willink\cite{wij00} provided constructions of such
sequences for binary data streams and studied their properties.
Bajic\etal\cite{baj03,baj04} showed that the distributed sequence entails
a simultaneous search for a set of synchronization words. Each word in the
set of sequences is required to be bifix-free. In addition there arises a new
requirement that no prefix of any length of any word in the set should be
a suffix of any other word in the set. This property of the set of
synchronization words was termed as \emph{cross-bifix-free} in
\cite{baj03,baj04,ste12}. In the same works, the properties of sets of words that are
cross-bifix-free were statistically analyzed. In this article we term
the set of words which are cross-bifix-free as a \emph{cross-bifix-free code.}
In the above example of a distributed sequence the set of words
$\{(1,1,0,0,0,0,0),(1,1,0,0,0,1,0),(1,1,0,1,0,0,0),$ $(1,1,0,1,0,1,0)\}$ forms
a cross-bifix-free code.

In a follow up work, Bajic
\cite{baj07} provided a new construction of cross-bifix-free codes over
a binary alphabet for word lengths up to eight. This specific construction
uncovers interesting connections to the Fibonacci sequence of numbers. In
particular, the number $S(n)$ of binary words of length $n$, for $3\le n\le8,$ which are
cross-bifix-free satisfies the Fibonacci recursion
$$
S(n) = S(n-1) + S(n-2).
$$
It is noted in \cite{baj07} that although this construction gives larger sets
compared to distributed sequences \cite{wij00} for $n\le8$, the sizes of the sets are
relatively smaller for lengths greater than eight. In a recent work
Bilotta\etal\cite{bil12}
introduced a new construction of binary cross-bifix-free codes based on
lattice paths, and showed that their construction attains greater cardinality compared
to the ones in \cite{baj07}.

In this work, we revisit the construction in Bajic\cite{baj07}. We give a new
construction of cross-bifix-free codes that  generalizes
the construction of \cite{baj07} in two ways. First, we provide new binary
codes that are greater in cardinality compared to the ones in \cite{bil12} for
larger lengths. In
the process we discover interesting connections of the size of the codes
obtained to
the so-called $k$-generalized Fibonacci numbers. Secondly, we generalize
the construction to $q$-ary alphabets for any $q, q\ge2.$ To the best of
our knowledge, this is the first construction of cross-bifix-free codes
over alphabets of size greater than two. The size of the
generalized $q$-ary constructions are also related to a Fibonacci sequence,
that we call the $(q-1)$-weighted $k$-generalized Fibonacci sequence (see
Section \ref{sec:notation} for the exact definition). Using this relation to the
Fibonacci sequences we analyze the asymptotic size of our construction.
In the process of this asymptotic analysis, we generalize a result of
Dresden \cite{dre11} on $k$-generalized Fibonacci sequence to $(q-1)$-weighted
$k$-generalized Fibonacci sequence. The main asymptotic result on the size
of cross-bifix-free codes that we prove is described in the theorem below.
\begin{thm}
\label{thm:main-thm}
Let $C(n,q)$ denote the maximum size of a cross-bifix-free code of
length $n$ over an alphabet of size $q$. Then,
\begin{align}
    \liminf_{n\to\infty} \frac{C(n,q)}{q^n/n} &\ge \frac{q-1}{qe}
                                              \simeq 0.368 \frac{q-1}q.\label{eq:lower}\\
        \limsup_{n\to\infty} \frac{C(n,q)}{q^n/n} &\le
                             \frac12 = 0.5 \label{eq:upper}.
\end{align}
\end{thm}

Note that the lower bound is within a constant factor of the best possible construction.
The ratio between the lower and the upper bound increases
towards $2/e=0.736$ for larger alphabet sizes.
In comparison, a similar ratio of the size of the
binary codes constructed by Bilotta\etal\cite{bil12} or
the distributed sequences by van Wijngaarden and Willink\cite{wij00},
to the quantity $2^n/n$, asymptotically goes to zero.

The rest of the paper is organized as follows. Our presentation
is provided for general alphabet size $q,\ q\ge 2$, and the results for the
binary alphabet are obtained as a special case. In Section
\ref{sec:construction}, we
provide the construction of the cross-bifix-free code and show that for the
binary alphabet it is
optimal for lengths $n\le14,$ barring an exception at $n=9.$ In Section
\ref{sec:near-optimality} we study the asymptotic behavior of the size of
cross-bifix-free codes obtained from our construction.
{In particular, we exhibit \eqref{eq:lower} and \eqref{eq:upper} in
Theorem \ref{prop:limit2} and Theorem \ref{thm:upper-bound}, respectively.}
Results on the
behavior of the $(q-1)$-weighted $k$-generalized Fibonacci sequence are
also presented in this section. Lengthy calculations and some proofs are
deferred to the Appendix. In the  section  below we
introduce the basic notations and definitions required.

\section{Notations and definitions}
\label{sec:notation}
Let $\ZZ_q = \{0,\dots,q-1\}$ be an alphabet of $q$ elements. We denote by
$\ZZ_q^*$ all the nonzero elements of the set $\ZZ_q$, that is, $\ZZ_q^*
= \ZZ_q\setminus\{0\}.$
A consecutive sequence of $m$ elements $b\in\ZZ_q$ is denoted by the short
form $b^m.$ As an example, the word $(0,0,1,1,1,0,1)$ is represented in
short as $(0^2,1^3,0,1)$. For convenience, if $m=0$ then $b^m$ is used to
denote the absence of any element.

\begin{defn}
For a word $\bfu\in\ZZ_q^n$, a word $\bfv$ is called a \emph{prefix} of
$\bfu$ if we can write $\bfu$ as $\bfu = (\bfv,\bfw)$, for some word $\bfw$.
The word $\bfw$ is called a \emph{suffix} of $\bfu\in\ff_q^n$ if we can
write $\bfu$ as $\bfu = (\bfv,\bfw)$, for some word $\bfv$.
\end{defn}
For any word $\bfu$ we only consider prefixes and suffixes which have
length strictly less than the length of $\bfu.$

\begin{defn}
A word $\bfu\in\ff_q^n$ is called \emph{bifix-free} if the prefix of any
length of the word is not a suffix of the word.
\end{defn}

\begin{defn}
A \emph{cross-bifix-free code} is a set of words in $\ZZ_q^n$ which satisfy the
property that the prefix of any length of any word is not the suffix of
any word in the set, including itself.
\end{defn}
We denote the maximum size of
a cross-bifix-free code by the notation $C(n,q)$.

\begin{defn}
The \emph{$(q-1)$-weighted $k$-generalized Fibonacci sequence} is a sequence of numbers which
satisfies the recurrence relation
$$
    F_{k,q}(n) = (q-1) \sum_{l=1}^k F_{k,q}(n-l),
$$
for some initial values of $F_{k,q}(0),\dots,F_{k,q}(k-1)$.
For $q=2$, the sequence obtained is called a \emph{$k$-generalized
Fibonacci sequence.} For $q=2$, $k=2$, and the initialization
$F_{2,2}(0) = 1, F_{2,2}(1) = 2,$ we obtain the usual Fibonacci
sequence.
\end{defn}
The $(q-1)$-weighted $k$-generalized Fibonacci sequence is a special case
of the \emph{weighted $k$-generalized Fibonacci sequence} which satisfies
the recurrence relation \cite{sch08,lev85}
$$
F_{k}(n) = a_1 F_{k}(n-1) + a_2 F_{k} (n-2) + \cdots + a_{k}
F_{k} (n-k),
$$
where the weights are given by $a_1,a_2,\dots,a_{k}\in \mathbb Z$, and
$\mathbb Z$ denotes the integers. Setting all the
weights equal to $q-1$ gives the sequence in the above definition.

The $(q-1)$-weighted $k$-generalized Fibonacci sequence arises in the
study of cross-bifix-free codes as described in the section below.

\section{A Construction of Cross-Bifix-Free Codes}
\label{sec:construction}
In this section we provide a general construction of cross-bifix-free codes
over the $q$-ary alphabet. Interestingly, the sizes of our construction are
related to the $(q-1)$-weighted $k$-generalized weighted Fibonacci numbers
$F_{k,q}(n)$.
The initialization on $F_{k,q}(n),\ 0\le n\le k-1$ that we use is given as
\begin{equation}
\label{eq:fibonacci-initial}
F_{k,q}(n)= q^n,\ n = 0,\dots,k-1.
\end{equation}

Below, we describe the family of cross-bifix-free codes in the space $\ZZ_q^n$. The
family is obtained by varying the value of $k$.\\

\noindent{\bf The construction:}
For any $2\le  k\le  n - 2$, denote by
$\cS_{k,q}(n)$ the set of all words $(s_1, s_2,\dots,s_n)$ in
$\ZZ_q^n$ that  satisfy the following two properties:
\begin{enumerate}
\item $s_1 = s_2 = \cdots = s_k = 0$, $s_{k+1}\ne 0$ and $s_{n}\ne 0$,
\item the subsequence $(s_{k+2},s_{k+3},\dots ,s_{n-1})$ does not contain
any string of $k$ consecutive $0$'s.
\end{enumerate}

This construction implies that $\cS_{k,q}(n)$ contains all possible words
of length $n$ that start with $k$ zeroes, end with a nonzero element, and
have at most $k-1$ consecutive zeroes in the last $n-1$ coordinates.
In the remaining part of this
section we show that for every $k,\ k=2,\dots,n-2$, this set of words
forms a cross-bifix-free code. We determine its size in terms of the
Fibonacci sequence.
First, in the theorem below, we show that $\cS_{k,q}(n)$ is
a cross-bifix-free code. Additionally, we show that
the code $\cS_{k,q}(n)$ has the property that it can not be expanded
while preserving the property that it is cross-bifix-free. That is, for
every word $\bfx\in\ff_q^n\setminus\cS_{k,q}(n)$, the set
$\{\bfx\}\cup\cS_{k,q}(n)$ is not cross-bifix-free.
\begin{thm}
For any $k,\ 2\le k\le n-2$, the set $\cS_{k,q}(n)$ is a nonexpandable
cross-bifix-free code.
\end{thm}

\begin{proof}
To see that $\cS_{k,q}(n)$ is a cross-bifix-free code, note that the prefix of
any word of $\cS_{k,q}(n)$ starts with $k$ consecutive zeroes. But in the last $n-1$
coordinates of any word, we have at most $k-1$ consecutive zeroes, and
the last coordinate is always nonzero. Thus,
no prefix of any length of any word can match any suffix of itself or of
any other word in $\cS_{k,q}(n)$.

To show that $\cS_{k,q}(n)$ is nonexpandable we consider all the possible
configurations of words that could be appended to the set $\cS_{k,q}(n)$.
First we note that we can not append any word starting with a nonzero
element since the nonzero element occurs in the last coordinate of some
word in $\cS_{k,q}(n)$. Similarly, we can not append any word ending with
a zero element. The other possible configurations of words that
we need to consider are as follows.
\begin{itemize}
    \item Let $\bfs$ be a word which contains at least $k$
        consecutive zeroes in the last $n-1$ coordinates. We consider the
        suffix in $\bfs$ that starts with the last set of $k$ consecutive zeroes and
        contains at most $k-1$ consecutive zeroes following it, that is,
        the suffix has the form $(0^k,\alpha,\bfu)$, where
        $\alpha$ is nonzero and $\bfu$ is a vector of length $m$ that has
        at most $k-1$ consecutive zeroes. Then the word of length $n$
        $(0^k,\alpha,\bfu,1^{n-m-k-1})$ is a word in
        $\cS_{k,q}(n)$ and has a prefix matching a suffix of $\bfs$. Thus
        $\bfs$ can not be appended to $\cS_{k,q}(n)$.

    \item Let $\bfs$ be a word which contains a prefix of at most $k-1$
        zeroes followed by a nonzero element, that is $\bfs
        = (0^l,\alpha,\bfu)$, where $0<l\le k-1$,
        $\alpha$ is nonzero, and $\bfu$ has length $n-l-1$. It is readily seen
        that $(0^l,\alpha)$ is also the suffix of the word
        $(0^k,1^{n-k-l-1},0^l,\alpha)$ in
        $\cS_{k,q}(n)$.
        Hence, such a word can not be appended to
        $\cS_{k,q}(n)$.
\end{itemize}
Thus, no additional word can be appended to the set $\cS_{k,q}(n)$ while
still preserving the cross-bifix-free property.
\end{proof}

The nonexpandability of the construction above parallels the
nonexpandability of the cross-bifix-free codes obtained in
Bajic \cite{baj07} and Bilotta\etal\cite{bil12}.  However, note that the nonexpandability does
not automatically indicate the optimality of the
construction, as is evident from the many values of $k$ for which the
nonexpandability holds true. In the following sections, we instead show
that the largest sized set
obtained by optimizing over the value of $k,\ k=2,\dots,n-2$, differs (in
ratio) from
the size of the optimal code by only a factor of a constant $2(q-1)/(qe)$.

We first describe a recursive construction of the set $\cS_{k,q}(n)$ in
terms of the sets $\cS_{k,q}(n-l),\ l=1,\dots,k.$ This recursive
construction immediately establishes the connection to the Fibonacci
recurrence and helps us determine the size of the set in terms of the
Fibonacci numbers.
\begin{thm}
\label{thm:recursion}
\begin{equation*}
\cS_{k,q}(n) =
\begin{cases}
    \big\{(0^k,\alpha, \bfs, \beta): \alpha, \beta \in \ZZ_q^*,
    \bfs\in\ff_q^{n-k-2}\big\},\\
    \qquad\qquad\qquad\qquad\qquad\quad k+2\le n\le 2k+1,\\
    \bigcup_{l=1}^k \big\{(\bfs,0^{l-1},\alpha): \bfs\in\cS_{k,q}(n-l),
    \alpha\in\ZZ_q^*\big\},\\
        \qquad\qquad\qquad\qquad\qquad\qquad\qquad\quad     2k+2\le n.
\end{cases}
\end{equation*}
\end{thm}

\begin{proof}
For $n=k+2,\dots,2k+1$, the coordinates
$s_{k+2},\dots,s_{n-1}$ necessarily have at most $n-k-2<k$ zeroes and hence
can contain all the words of length $n-k-2.$ This establishes
the result for $n=k+2,\dots,2k+1.$

Now, let $n\ge 2k+2.$
For brevity denote each set on the right-hand side (RHS) of the
equation in Theorem
\ref{thm:recursion} by
\begin{equation}
\label{eq:tau}
\mathcal{T}_l(n)\triangleq\{(\bfs,0^{l-1},\alpha):\bfs\in\cS_{k,q}(n-l),\alpha\in\ZZ_q^*\}.
\end{equation}
Note that the sets {$\mathcal{T}_l(n)$}
are mutually disjoint for different $l$
since the last {$l$} coordinates have different structure for the different
sets. To show that {$\cS_{k,q}(n)\subseteq\cup_l \mathcal{T}_l(n)$},
note that any element $\bfu\in\cS_{k,q}(n)$ has at most $k-1$ zeroes in the last $n-1$
coordinates and hence the word $\bfu$ must be of the form $\bfu
= (0^k,u_{k+1},\dots,u_{n-l-1},0^{l-1},\alpha)$ where
$u_{n-l-1},\alpha\in\ZZ_q^*$ and $l\in\{1,\dots,k\}$. Thus,
{$\bfu\in\mathcal{T}_l(n)$.}

To show the reverse inclusion, let $l\in\{1,\dots,k\}$ and let
$\bfs\in\cS_{k,q}(n-l)$. Note that $\bfs$ ends with a nonzero element.
The word $(\bfs,0^{l-1},\alpha)$ where $\alpha\in\ZZ_q^*$, starts with
a sequence $0^k$, ends with a nonzero element and has at most $k-1$
consecutive zeroes in the last $n-1$ coordinates. Hence
$(\bfs,0^{l-1},\alpha)\in\cS_{k,q}(n)$ and the set
$\{(\bfs,0^{l-1},\alpha): \alpha\in\ZZ_q^*\}$
is a subset of $\cS_{k,q}(n)$. Hence, {$\mathcal{T}_l(n)\subseteq\cS_{k,q}(n)$}
for every $l=1,\dots,k.$
\end{proof}

\begin{cor}
\label{cor:cardinality}
The cardinality of  $\cS_{k,q}(n)$ for $n\ge3$ is given by the equation
\begin{equation*}
S_{k,q}(n) = |\cS_{k,q}(n)| =(q-1)^2F_{k,q}(n-k-2).
\end{equation*}
\end{cor}

\begin{proof}
For $n=k+2,\dots,2k+1,$ the corollary can be readily verified from the
expression in \eqref{eq:fibonacci-initial} and Theorem
\ref{thm:recursion}. We use an induction argument for $n\ge 2k+2$. Assume
that the corollary is true for $n<N$ where $N\ge2k+2.$
First, note that by using the definition in \eqref{eq:tau}, we get
\begin{align*}
\sum_{l=1}^k|\mathcal{T}_l(N)| &=
    \sum_{l=1}^k\left|\{(\bfs,0^{l-1},\alpha):\bfs\in\cS_{k,q}(N-l),\alpha\in\ZZ_q^*\}\right|\\
    &= \sum_{l=1}^k (q-1) S_{k,q}(N-l).
\end{align*}
Now,
\begin{align*}
S_{k,q}(N) &= {\sum_{l=1}^k \big|\mathcal{T}_l(N)\big|}\\
            &= \sum_{l=1}^k (q-1) S_{k,q}(N-l) \\
            &= \sum_{l=1}^k (q-1) (q-1)^2 F_{k,q}(N-l-k-2)\\
            &= (q-1)^2 F_{k,q} (N-k-2).
\end{align*}
We used the induction argument in the second last step. This proves the
corollary.
\end{proof}

For fixed $n$ and $q$, the largest size of the set $\cS_{k,q}(n)$ can be
obtained by optimizing over the choice of $k.$ Let $S(n,q)$ denote this
maximum. It is given by the expression
\begin{equation}
\label{eq:snq}
S(n,q) = \max\{(q-1)^2 F_{k,q}(n-k-2): 2\le k\le n-2\}.
\end{equation}

In particular, the size $S(n,q)$ is upper bounded by the maximum
cardinality $C(n,q)$ of a cross-bifix-free code.

\subsection{Sizes of cross-bifix-free codes for small lengths}
The size of binary cross-bifix-free codes obtained in
Bilotta\etal\cite{bil12} is obtained by counting lattice paths, in
particular, Dyck
paths.
\begin{thm}[Bilotta\etal\cite{bil12}]
\label{thm:bilotta}
{Let $B(n)$ denote the size of a binary cross-bifix-free code of length $n$ constructed by Bilotta\etal\cite{bil12}.}
For $m\ge1$, let $C_m = \frac1{m+1}\binom{2m}m$ denote the $m$-th Catalan
number. Then
\begin{multline*}
    B(n) = \\\begin{cases}
C_m, & n = 2m+1, m\ge1,\\
\displaystyle\sum_{i=0}^{m/2} C_iC_{m-i}, &n = 2m+2,m\text{ odd},\\
\displaystyle\sum_{i=0}^{(m+1)/2} C_iC_{m-i}-C_{(m-1)/2}^2, &n = 2m+2,m\text{ even}.
\end{cases}
\end{multline*}
\end{thm}

For values of $n \le 16,$ it is verified by numerical computations that the
sizes obtained by our construction are all optimal, except for the
value $n=9.$ In particular we get the Table \ref{tab:table} of values for $3\le
n\le 30.$ The first column gives the value of the word length $n$, the
second column shows the sizes of the codes obtained in
Bilotta\etal\cite{bil12},
the third column gives the sizes obtained from our construction after
optimizing over different values of $k$. Finally, the last column gives the
values of $k$ for which $S_{k,q}(n)$ achieves the maximal size in the third
column. The numbers in bold denote the sizes that are known to be optimal.

\begin{table}[!h]
\caption{Table comparing the values from \cite{bil12} with
\eqref{eq:snq}}
\centering
\begin{tabular}{|cccc|}
    \hline
$n$& \cite{bil12}  & Eq.~\eqref{eq:snq} & $k$ \\    \hline
3  & \textbf{1 }       & \textbf{1}       & - \\
4  & \textbf{1 }       & \textbf{1}       & 2 \\
5  & \textbf{2 }       & \textbf{2}       & 2 \\
6  & \textbf{3 }       & \textbf{3}       & 2 \\
7  & \textbf{5 }       & \textbf{5}       & 2 \\
8  & \textbf{8 }       & \textbf{8}       & 2 \\
9  & \textbf{14}       & 13      & 2\\
10 & 23       &  \textbf{24 }     & 3 \\
11 & 42       &  \textbf{44 }     & 3 \\
12 & 72       &  \textbf{81 }     & 3 \\
13 & 132      &  \textbf{149}     & 3 \\
14 & 227      &  \textbf{274}     & 3 \\
15 & 429      &  \textbf{504}     & 3 \\
16 & 760      & \textbf{927}      & 3 \\ \hline
\end{tabular}
\hspace{0.1in}
\begin{tabular}{|cccc|}
    \hline
$n$& \cite{bil12}  & Eq.~\eqref{eq:snq} & $k$ \\    \hline
17 & 1430     & 1705              & 3 \\
18 & 2529     & 3136              & 3 \\
19 & 4862     & 5768              & 3\\
20 & 8790     & 10671   & 4 \\
21 & 16796    & 20569   & 4 \\
22 & 30275    & 39648   & 4 \\
23 & 58786    & 76424   & 4 \\
24 & 107786   & 147312  & 4 \\
25 & 208012   & 283953  & 4 \\
26 & 380162   & 547337  & 4 \\
27 & 742900   & 1055026 & 4 \\
28 & 1376424  & 2033628 & 4 \\
29 & 2674440  & 3919944 & 4 \\
30 & 4939443  & 7555935 & 4\\ \hline
\end{tabular}
\label{tab:table}
\end{table}

The optimality of the values for $n\le 16$ is proved computationally by
setting up a specific program that searches for the largest
clique in a graph. The graph consists of vertices which correspond to the
set of all words in $\ff_2^n$ that are bifix-free. An edge exists between
two vertices, i.e., two words, if they are mutually cross-bifix-free. The
algorithm \verb|MaxCliqueDyn|\cite{maxcliquedyn} is used to determine the
maximum size of the clique in the graph. This algorithm shows that the
values denoted by bold in Table \ref{tab:table} are optimal.

Note that our construction has larger size
than the construction in \cite{bil12} for all values of $n,\ 13\le n\le30.$
This trend is observed asymptotically too, as we describe in the following
sections.

\section{Near optimality of the size $S(n,q)$}
\label{sec:near-optimality}
In this section we show that the size $S(n,q)$ is close to the maximum size
$C(n,q)$.
The ratio $S(n,q)/C(n,q)$ measures how close the construction in Section
\ref{sec:construction} is to the optimal value. The following theorem
gives an asymptotic lower bound on this ratio.
\begin{thm} \label{prop:limit}
The following limit holds:
\begin{equation}\label{eqn:limit}
\liminf_{n\to\infty} \frac{S(n,q)}{C(n,q)}\ge 2\frac{q-1}{qe}.
\end{equation}
\end{thm}
This lower bound is proved by showing a lower bound on $S(n,q)$ and an
upper bound on $C(n,q)$.
The derivation of the lower bound on $S(n,q)$ crucially depends on the
properties of the $(q-1)$-weighted $k$-generalized Fibonacci sequence of
numbers. We digress in the next subsection to first establish these needed
properties.

\subsection{Properties of the Fibonacci sequence $F_{k,q}(n)$}\label{sec:fib}
Levesque \cite{lev85} showed in a very general context that to every weighted
$k$-generalized Fibonacci sequence of numbers we can associate
a \emph{characteristic polynomial} (see Theorem~\ref{thm:lev} in the
Appendix). For the $(q-1)$-weighted $k$-generalized Fibonacci sequence,
this polynomial specializes to the following form
\begin{equation}
\label{eq:fibonacci-poly}
f(x) = x^k - (q-1) \sum_{i=0}^{k-1} x^i.
\end{equation}
Below, we state the properties of this polynomial and of the corresponding
Fibonacci numbers. The initialization sequence that we use is the one
described in \eqref{eq:fibonacci-initial}. The proofs in
this section are omitted for clarity of presentation and are instead
provided in Appendix.

\begin{prop}
\label{prop:fx}
The polynomial $f(x)$ has distinct roots with a unique real root
$\alpha \equiv \alpha(k,q)$ outside the unit circle.
The root $\alpha$ lies in the interval $(1,q)$.
\end{prop}
The value of the root $\alpha$ is in fact close to $q$. An estimate of this
root is given by the following lemma.
\begin{lem} \label{lem:beta}
There exists a number $K_q$ such that the following holds. For all $k\ge K_q$,
there exists a $\beta \equiv\beta(k,q)$ in the interval $(q-\frac 1{q^{k-1}},q)$ such that
\begin{equation}\label{eqn:beta}
q-\frac {q-1}{\beta^k}<\alpha<q.
\end{equation}
\end{lem}
Finally,  the Fibonacci numbers can be expressed in terms of this real root
$\alpha$. Let $[x]$ denote the integer closest to  the real number $x.$
\begin{prop}\label{prop:fib}
Let $q\ge2$. The $n$-th number in the $(q-1)$-weighted $k$-generalized Fibonacci
sequence is given by the expression
\begin{equation*}
F_{k,q}(n) = \left[
    \frac{(\alpha-1)\alpha^{n+1}}{(q+(k+1)(q-\alpha))(q-1)}
    \right].
\end{equation*}
\end{prop}
We note here that
Proposition \ref{prop:fx} is a generalization to $q\ge3$ of the result obtained by
Miles \cite{mil60} for $q=2.$ We adopt a technique similar to the one in Miller
\cite{mil71}. Additionally,
Proposition \ref{prop:fib} is a generalization of the result in
Dresden\cite{dre11} to $(q-1)$-weighted $k$-generalized Fibonacci numbers, for
$q\ge3.$ For $q=2,$ the expression above reduces to the expression for the
sequence $F_{k,2}(n)$ as obtained in \cite{dre11}.

\subsection{A Lower Bound on $S(n,q)$}
Using the properties of the Fibonacci numbers from the previous subsection,
we establish an asymptotic lower bound on the size $S(n,q)$.
\begin{thm} \label{prop:limit2}
The asymptotic size $S(n,q)$ satisfies the limit,
\begin{equation}\label{eqn:limit2}
\liminf_{n\to\infty} \frac{S(n,q)}{q^n/n}\ge \frac{q-1}{qe}.
\end{equation}
\end{thm}
\begin{proof}
Using Corollary \ref{cor:cardinality} and Proposition \ref{prop:fib} in
successive steps we obtain,
\begin{align*}
\frac{S(n,q)}{q^n/n} & \ge \frac{n(q-1)^2F_{k,q}(n-k-2)}{q^n}\\
& \ge \frac{n(q-1)^2\left(\frac{(\alpha-1)(\alpha^{n-k-1})}{(q+(k+1)(q-\alpha))(q-1)}-1\right)}{q^{n}} \\
& \ge  \left(\frac {q-1}{q}\right) \left(\frac
    {\alpha-1}{\alpha}\right)\left(\frac{\alpha}{q}\right)^n
    \left(\frac{n}{\alpha^k}\right) - o(1),
\end{align*}
where the term $o(1)\to0$ as $n\to\infty.$
To derive the asymptotics we choose $n$ as an increasing function of $k$:
$$
n \equiv n(k) = \ceiling{c \alpha^k},
$$
where $c$  is a positive constant. Note that $\alpha$ is also
a function of $k.$ We obtain,
\begin{align*}
\frac{S(n,q)}{q^n/n}& \ge  \left(\frac {q-1}{q}\right) \left(\frac
    {\alpha-1}{\alpha}\right)\left(\frac{\alpha}{q}\right)^{\ceiling{c\alpha^k}}
    \left(\frac{\ceiling{c\alpha^k}}{\alpha^k}\right) - o(1)\\
& \ge  \left(\frac {q-1}{q}\right) \left(\frac {\alpha-1}{\alpha}\right)\left(\frac{\alpha}{q}\right)^{(c\alpha^k+1)}\cdot c - o(1) \\
& =  \left(\frac {q-1}{q}\right) \left(\frac
    {\alpha-1}{q}\right)c\left(\frac{\alpha}{q}\right)^{c\alpha^k}
    - o(1).
\end{align*}
The last term in the right-hand side (RHS) of the equation above can be further lower bounded
by using Lemma~\ref{lem:beta}.
We assume that there exists a number $K_q$
and $k\ge K_q$, as required by the lemma.
\begin{align*}
c\left(\frac{\alpha}{q}\right)^{c\alpha^k}
& \ge c\left(1-\frac{q-1}{q\beta^k}\right)^{cq^k}
\\
& =c\left(\left(1-\frac{q-1}{q\beta^k}\right)^{c\beta^k}\right)^{\left(q/{\beta} \right)^k}.
\end{align*}
The RHS of the above equation tends to
$
ce^{-\frac{c(q-1)}{q}}
$ as $k\to\infty$
since $\left(1+\frac{1}{x}\right)^x\to e$ as $x\to\infty$ and $\beta(k,q)\to q$ as $k\to \infty$.
The term $ce^{-\frac{c(q-1)}{q}}$ attains a maximum of $\frac{q}{(q-1)e}$ when $c=q/(q-1)$. The
theorem follows by substituting this value into the lower bound above.
\end{proof}

\subsection{An upper bound on the maximum size $C(n,q)$}
Let $M$ denote the size of a cross-bifix-free code of length $n$ over an
alphabet of size $q$.
An upper bound for the maximum size of a cross-bifix-free code is readily
obtained from the study of the statistical properties of such sets in the
data stream. The main object of study is the time when the search for
any word of the cross-bifix-free code in the data stream returns with
a positive match.
Bajic\etal\cite{baj03,baj04} establish the probability distribution
function of this time, the expected time duration for a match, and the
variance of this distribution. The variance $\sigma^2$ of the time for the
first match is given by the expression \cite[eq.~(18)]{baj04}
\begin{equation}
\label{eq:variance}
\sigma^2 = (1-2n)\frac{q^n}M + \frac{q^{2n}}{M^2}.
\end{equation}
Using the property that the variance is always nonnegative immediately
gives us the required upper bound on any cross-bifix-free code.
In particular, we have the theorem
\begin{thm}
\label{thm:upper-bound}
Let $C(n,q)$ denote the maximum size of a cross-bifix-free code in
$\ZZ_q^n$. Then,
$$C(n,q)\le \frac{q^n}{2n-1}.$$
\end{thm}
We remark that this upper bound, albeit immediate from
\eqref{eq:variance}, was not noted
in the previous works on the size of the cross-bifix-free codes.
Combining Theorem \ref{prop:limit2} and Theorem \ref{thm:upper-bound}, we obtain
Theorem~\ref{prop:limit} {and Theorem \ref{thm:main-thm}}.

\subsection{Comparison to earlier results}\label{subsec:comparison}

{To compare the construction in this work with the new construction of
binary cross-bifix-free codes \cite{bil12} and the construction of
distributed sequences \cite{wij00},
we study the asymptotics of their respective constructions
for large $n$. In both cases, we exhibit that the size of the previous constructions
is a negligible fraction of $2^n/n$,
 in contrast to the nearly optimal construction described in the previous
section.}

The asymptotic behavior of the construction in \cite{bil12} is obtained
from the expressions in Theorem~\ref{thm:bilotta}. We obtain that
$$
B(n) \begin{cases}
    = C_m, & n = 2m+1,\\
    \le C_{m+1}, & n = 2m+2,
\end{cases}
$$
where $C_m = \frac1{m+1}\binom{2m}m$ is the $m$-th Catalan number and
$m\ge1.$
Using Stirling's approximation, we get that the
number $C_m$ is approximately,
$$
C_m \simeq \frac1{m+1}\frac{2^{2m}}{\sqrt{\pi m}}.
$$
Thus for $n$ odd,
$$\frac{B(n)}{2^n/n} =  \frac{(2m+1)C_m}{2^{2m+1}} \simeq
\frac{2m+1}{2(m+1)}\frac{1}{\sqrt{\pi m}},$$
which goes to zero as $2m+1 = n \to \infty$. Similar conditions hold for the case $n=2m+2.$ Thus the
construction in \cite{bil12} is a negligible fraction of $2^n/n$.

{On the other hand,
van Wijngaarden and Willink \cite[Eq.~(4)]{wij00} showed that for
a set of distributed sequences of length $n$, and with
$h$ synchronisation positions,
\begin{equation}\label{eq:distseq}
n\le \floor{h^2/4}+1.
\end{equation}
Let $D(n)$ denote the maximum size of a set of distributed sequences.
Then it follows from \eqref{eq:distseq} that
\begin{equation*}
D(n)\le 2^{n-h}\le 2^{n-2\sqrt{n-1}}.
\end{equation*}
Hence, the ratio $D(n)/(2^n/n)$ tends to zero with increasing $n$.
}

\section{Conclusion}
We provided a new construction of cross-bifix-free codes that are close to
the maximum possible size. The construction for the binary codes is shown
to be larger than the previously constructed codes for all lengths
$n\le30,$ barring an exception at $n=9.$ We also provided the first
construction of $q$-ary cross-bifix-free  codes for $q>2.$
In the process,
we established new results on the Fibonacci sequences, generalizing some
earlier works on these sequences.

\appendix\section{Appendix}
\label{sec:appendix}

In this appendix, we provide the proofs of the results on the Fibonacci
sequences that are stated in Section \ref{sec:fib}. First, we recall a very
general theorem on weighted $k$-generalized Fibonacci sequences proved in
Levesque\cite{lev85}.
\begin{thm}[Levesque\cite{lev85}]
\label{thm:lev}
Let $k\ge 2$. Let $F_k(n)$ be defined by the following recurrence relation,
\begin{equation*}
    F_k(n)=\sum_{i=1}^k a_iF_k(n-i), \mbox{ for $n\ge k$},
\end{equation*}
for $a_i \in \ZZ,\ i=1,\dots,k,$
and with the initial conditions, $F_k(0),F_k(1),\ldots, F_k(k-1)$.
Additionally, suppose that the \emph{characteristic polynomial} $h(x)$ associated with the sequence
$\{F_k(n)\}_{k=0}^{\infty}$,
\begin{equation*}
    h(x)=x^k-\sum_{i=0}^{k-1} a_{k-i} x^i,
\end{equation*}
has distinct roots $\gamma_1,\gamma_2,\ldots,\gamma_k$. Then, for $n\ge k$,
the values of $F_k(n)$ are given by the expression
$$
F_k(n) =\sum_{j=0}^{k-1} p_{n-j}v_j,
$$
where,
\begin{align*}
v_0 & = u(0),\\
v_j & = u(j)-\sum_{i=1}^j a_i u(j-i), \quad\mbox{ for $1\le j\le k-1$},\\
p_j  &= \sum_{i=1}^k \frac{\gamma^{k-1+j}_i}{h'(\gamma_i)},\quad \mbox{for $j\ge 1$}.
\end{align*}
\end{thm}
For $a_i = q-1,\ i=1,\dots,k$, we obtain the corresponding expressions for
the $(q-1)$-weighted $k$-generalized Fibonacci numbers. In particular, the
polynomial $h(x)$ reduces to the polynomial $f(x)$ defined in
\eqref{eq:fibonacci-poly}.

We proceed with the proofs of the propositions in Section~\ref{sec:fib}. In
order to prove Proposition~\ref{prop:fx}, we first establish two lemmas
below. Define a polynomial $g(x)$ as
\begin{equation}
\label{eq:gx}
g(x) \triangleq   (x-1)f(x)=x^k(x-q)+(q-1).
\end{equation}
\begin{lem} \label{lem:fib1}
The polynomial $f(x)$ has a real root  in the interval $(1,q)$.
\end{lem}

\begin{proof} This follows from the fact that $f(1)=1-k(q-1)<0$ and $f(q)=g(q)/(q-1)=1>0$.
\end{proof}

\begin{lem}\label{lem:fib2}
Let $\alpha \equiv \alpha(k,q)$ be the real root of $f(x)$ in $(1,q)$.
Then the polynomial $g(x)$, and consequently the polynomial $f(x)$,
satisfies the following inequalities.
\begin{enumerate}
\item $g(x)>0$ for $x\in (\alpha,\infty)$,
\item $g(x)<0$ for $x\in (1,\alpha)$.
\end{enumerate}
\end{lem}

\begin{proof}
Observe that
\begin{equation*}
g'(x)=x^{k-1}((k+1)x-kq),
\end{equation*}
\noindent and so $g'(x)<0$ for $x\in [1,kq/(k+1))$ and $g'(x)>0$ for  $x\in (kq/(k+1)),q]$.
Since $g(1)=g(\alpha)=0$, $g(kq/(k+1))<0$ and $g(q)>0$, the lemma follows.
\end{proof}
\vspace{2mm}
Next, we establish Proposition \ref{prop:fx}.\\
\begin{proof}[Proof of Proposition \ref{prop:fx}]
First, we show that the roots of $g(x)$ in \eqref{eq:gx}, and hence of $f(x)$, are distinct.
Indeed, $g'(x)=0$ if and only if $x=0$ or $x=\frac{kq}{k+1}$.
However, $g(0)\ne 0$ and $g\left(\frac{kq}{k+1}\right)\ne 0$. Therefore, the roots are distinct.

Next, let $\gamma$ be a root of $f(x)$ with $\gamma\ne\alpha$. We prove by contradiction that $|\gamma|<1$.
We consider the two cases $|\gamma|>|\alpha|$ and $|\gamma|<|\alpha|$
separately.
Suppose $|\gamma|>|\alpha|$. Since $\gamma^k=(q-1)\sum_{i=1}^{k-1}
\gamma^i$, we get
\begin{equation*}
|\gamma|^k=|(q-1)\sum_{i=0}^{k-1} \gamma^i|\le (q-1)\sum_{i=0}^{k-1}|\gamma|^i,
\end{equation*}
\noindent and so, $f(|\gamma|)\le 0$, contradicting part (i) of Lemma \ref{lem:fib2}.

Next, suppose $|\gamma|\in [1,|\alpha|]$. Since $\gamma$ is also a root of
$g(x)$, $q\gamma^k=\gamma^{k+1}+(q-1)$. Then
\begin{equation}\label{eqn:gamma}
q|\gamma|^k=|\gamma^{k+1}+(q-1)|\le |\gamma|^{k+1}+(q-1),
\end{equation}
which implies that $g(|\gamma|)\ge 0$. Then by part (ii) of Lemma
\ref{lem:fib2}, $|\gamma|\in\{1,|\alpha|\}$ and equality in
(\ref{eqn:gamma}) holds.
Hence, $\gamma^{k+1}$ and $\gamma^k$ are real, implying that $\gamma$ is
real. Since the roots of $g(x)$ are distinct,
$\gamma\in\{-1,-\alpha\}$. But $g(-1)=(-1)^{k+1} (q+1)-(q-1)\ne 0$ and
$g(-\alpha)=2\alpha^kq+(1+(-1)^k)(q-1)\ne 0$, contradicting the fact that
$\gamma$ is a root of $g(x)$.

Therefore, $\alpha$ is the only root outside the unit circle. By Lemma \ref{lem:fib1}, $\alpha$ is in the required interval.
\end{proof}

\begin{proof}[Proof of Proposition \ref{prop:fib}]
We apply Lemma \ref{thm:lev} with $a_i = q-1,\ i=1,\dots,k$ and with
$h(x)=f(x)$.
Observe that $v_0=v_1=\cdots=v_{k-1}=1$. We obtain,
\begin{align*}
F_{k,q}(n) &= \sum_{j=0}^{k-1} p_{n-j}\\
&= \sum_{j=0}^{k-1} \sum_{i=1}^k \frac{\gamma_i^{k-1+n-j}}{f'(\gamma_i)}\\
&=\sum_{i=1}^k \frac{\gamma_i^{n}}{f'(\gamma_i)}\sum_{j=0}^{k-1}{\gamma_i^j}\\
&=\sum_{i=1}^k \left(\frac{\gamma_i^{n}}{f'(\gamma_i)}\right)\left(\frac{\gamma_i^k}{q-1}\right)
\\
&=\sum_{i=1}^k \left(\frac{\gamma_i^{n}(\gamma_i-1)}{(q+(k+1)(\gamma_i-q))\gamma_i^{k-1}}\right)\left(\frac{\gamma_i^k}{q-1}\right) \\
&=\sum_{i=1}^k \frac{\gamma_i^{n+1}(\gamma_i-1)}{(q-1)(q+(k+1)(\gamma_i-q))}.
\end{align*}
To obtain the fourth step we used the fact that $\gamma_i$'s are roots of
$f(x)$. Without loss of generality,
let $\gamma_1$ be the $\alpha(k,q)$ defined in Proposition \ref{prop:fib} and so
$|\gamma_i|<1$ for $i\ge 2$. We note that $q+(k+1)(\gamma_i - q)
= -kq+(k+1)\gamma_i$. We get the following sequence of inequalities for $q\ge3$:

\begin{align*}
& \left|\sum_{i=2}^k \frac{\gamma_i^{n+1}(\gamma_i-1)}{(q-1)(q+(k+1)(\gamma_i-q))}\right|\\
& \le \sum_{i=2}^k \left|\gamma^{n+1}_i\right|\left|\frac{(\gamma_i-1)}{(q-1)(q+(k+1)(\gamma_i-q))}\right|\\
& \le \sum_{i=2}^k \frac{|\gamma_i-1|}{(q-1)|kq-(k+1)\gamma_i|}\\
& < \sum_{i=2}^k \frac1{q-1} \frac{2}{k(q-1)-1} \\
& < \frac12,
\end{align*}
where the second last step is obtained by observing that for $i\ge2$ we have the inequalities $|\gamma_i-1|<2$ and $|kq-(k+1)\gamma_i| > kq - (k+1)$. The last step is obtained by applying the inequality $q\ge3.$
This completes the proof for $q\ge3$. The proof for $q=2$ is present in
Dresden\cite{dre11}.
\end{proof}

\vspace{2mm}

Finally, we prove Lemma \ref{lem:beta}. Again, for brevity, we denote
$\alpha(k,q)$ and $\beta(k,q)$ by $\alpha$ and $\beta$ respectively.\\
\begin{proof}[Proof of Lemma \ref{lem:beta}]
Observe that $g(q-1/q^{k-1})<0$ if and only if
\begin{equation*}
(1-1/q^k)^k>(1-1/q),
\end{equation*}
where $g(x)$ is the polynomial defined in \eqref{eq:gx}.
Since  $(1-1/q^k)^k \to 1$ as $k\to \infty$, there exists a constant $K_q$ such that $g(q-1/q^{k-1})<0$ for all $k\ge K_q$.

Hence, for all $k\ge K_q$, there exists $\beta$ in the interval $(q-1/q^{k-1},q)$ such that $g(\beta)<0$.
We claim that $\beta$ satisfies (\ref{eqn:beta}) by showing that $g(\beta)<0$ implies that $g(q- 1/{\beta^k})<0$.
Indeed, since
$$
                        g(\beta) =\beta^k(\beta-q)+(q-1)  < 0,
$$
we get
\begin{align*}
&&\beta                                                         & < q-\frac {q-1}{\beta^k}\\
\Rightarrow && g\left(q-\frac {q-1}{\beta^k}\right)&=\left(q-\frac
    {q-1}{\beta^k}\right)^k\left(-\frac{q-1}{\beta^k}\right)+(q-1)\\
            && &<0.
\end{align*}
Since $g(\alpha)=0$, we get
$
q-\frac {q-1}{\beta^k}<\alpha<q.
$
\end{proof}

\section*{Acknowledgement}
We are grateful to  Dragana Bajic for providing us with a copy of her work
\cite{baj07}, {and the anonymous reviewers for their helpful comments.
In particular, the comparison with distributed sequences in Subsection \ref{subsec:comparison}
was prompted by a perspicacious reviewer.}





\begin{IEEEbiographynophoto}{Yeow Meng Chee}
(SM~'08) received the B.Math. degree in computer science and combinatorics and
optimization and the M.Math. and Ph.D. degrees in computer science, from the University of Waterloo, Waterloo, ON, Canada, in 1988, 1989,
and 1996, respectively.

Currently, he is an Associate Professor at the Division of Mathematical
Sciences, School of Physical and Mathematical Sciences, Nanyang
Technological University, Singapore. Prior to this, he was Program Director
of Interactive Digital Media R\&D in the Media Development Authority of
Singapore, Postdoctoral Fellow at the University of Waterloo and IBM's
Z{\"u}rich Research Laboratory, General Manager of the Singapore Computer
Emergency Response Team, and Deputy Director of Strategic Programs at the
Infocomm Development Authority, Singapore. His research interest lies in
the interplay between combinatorics and computer science/engineering,
particularly combinatorial design theory, coding theory, extremal set
systems, and electronic design automation.
\end{IEEEbiographynophoto}
\begin{IEEEbiographynophoto}{Han Mao Kiah}
received the B.Sc.(Hon) degree in mathematics from the National
University of Singapore, Singapore in 2006.  Currently, he is working
towards his Ph.D. degree at the Division of Mathematical Sciences, School of
Physical and Mathematical Sciences,  Nanyang Technological University, Singapore.

His research interest lies in the application of combinatorics to
engineering problems in information theory. In particular, his interests
include combinatorial design theory, coding theory and power line
communications.
\end{IEEEbiographynophoto}
\begin{IEEEbiographynophoto}{Punarbasu Purkayastha}
received the B.Tech. degree in electrical engineering from Indian
Institute of Technology, Kanpur, India in 2004, and the Ph.D. degree in
electrical engineering from University of Maryland, College Park, U.S.A.,
in 2010.

Currently, he is a Research Fellow at the Division of Mathematical
Sciences, School of Physical and Mathematical Sciences,  Nanyang
Technological University, Singapore. His research interests include coding
theory, combinatorics, information theory and communication theory.
\end{IEEEbiographynophoto}
\begin{IEEEbiographynophoto}{Chengmin Wang}
received the B.Math. and Ph.D. degrees in mathematics from Suzhou
University, China in 2002 and 2007, respectively.

Currently, he is an Associate Professor at the School of Science,
Jiangnan University, China. Prior to this, he was a Visiting Scholar
 at the School of Computing, Informatics and Decision
Systems Engineering, Arizona State University, USA, from 2010 to
2011 and was a Research Fellow at the Division of Mathematical
Sciences, School of Physical and Mathematical Sciences, Nanyang
Technological University, Singapore from 2011 to 2012. His research
interests include combinatorial design theory and its applications
in coding theory and cryptography.
\end{IEEEbiographynophoto}
\vfill

\end{document}